\documentclass[conference]{IEEEtran}
\newtheorem{theorem}{Theorem}
\newtheorem{lemma}{Lemma}

\newtheorem{definition}{Definition}

\newtheorem{claim}{Claim}
\newtheorem{example}{Example}

\usepackage{cite}
\usepackage{verbatim}

\usepackage[cmex10]{amsmath}
\usepackage{amssymb}

\usepackage[caption=false,font=footnotesize]{subfig}
\usepackage{graphicx}

\usepackage{fixltx2e}
\usepackage{algorithmic,algorithm}
\begin{document}

\title{Multiterminal Source Coding with an Entropy-Based Distortion Measure}
%
%
%


\author{\IEEEauthorblockN{Thomas A. Courtade and Richard D. Wesel}
\IEEEauthorblockA{Department of Electrical Engineering\\
University of California, Los Angeles\\
Los Angeles, California 90095\\
Email: tacourta@ee.ucla.edu; wesel@ee.ucla.edu}}

\maketitle

\begin{abstract}
In this paper, we consider a class of multiterminal source coding problems, each subject to distortion constraints computed using a specific, entropy-based,  distortion measure.  We provide the achievable rate distortion region for two cases and, in so doing, we demonstrate a relationship between the lossy multiterminal source coding problems with our specific distortion measure and (1) the canonical Slepian-Wolf lossless distributed source coding network, and (2) the Ahlswede-K\"{o}rner-Wyner source coding with side information problem in which only one of the sources is recovered losslessly.
\end{abstract}


\section{Introduction}\label{sec:Intro}
\subsection{Background}
A complete characterization of the achievable rate distortion region for the classical lossy multiterminal source coding problem depicted in Fig. \ref{fig:MTSC} has remained an open problem for over three decades.  Several special cases have been solved:
\begin{itemize}
\item The lossless case where $D_x=0,D_y=0$.  Slepian and Wolf solved this case in their seminal work\cite{bib:SW}.
\item The case where one source is recovered losslessly: i.e., $D_x=0,D_y=D_{max}$.  This case corresponds to the source coding with side information problem of Ahlswede-K\"{o}rner-Wyner \cite{bib:AK},\cite{bib:W}.
\item The Wyner-Ziv case \cite{bib:W} where $Y^n$ is available to the decoder as side information and $X^n$ should be recovered with distortion at most $D_x$.
\item The Berger-Yeung case (which subsumes the previous three cases) \cite{bib:BY} where $D_x$ is arbitrary and $D_y=0$.
\end{itemize}

Despite the apparent progress, other seemingly fundamental cases, such as when $D_x$ is arbitrary and $D_y=D_{max}$, remain unsolved except perhaps in very special cases.

\subsection{Our Contribution}
In this paper, we give the achievable rate region for two cases subject to a particular choice of distortion measure $d(\cdot)$, defined in Section \ref{sec:probStatement}.  Specifically, for our particular choice of $d(\cdot)$, we give the achievable rate distortion region for the following two cases:
\begin{itemize}
\item The situation when $X$ and $Y$ are subject to a joint distortion constraint given a reproduction $\hat{Z}$:
\begin{align*}
\mathbb{E}\left[d(X,Y,\hat{Z}) \right]\leq D.
\end{align*}
\item The case where $X$ is subject to a distortion constraint given a reproduction $\hat{V}$:
\begin{align*}
\mathbb{E}\left[d(X,\hat{V}) \right]\leq D_x,
\end{align*}
and there is no distortion constraint on the the reproduction of $Y$ (i.e., $D_y$=$D_{max}$).
\end{itemize}

The regions depend critically on our choice of $d(\cdot)$, which can be interpreted as a natural measure of the soft information the reproduction $\hat{Z}$ symbol provides about the source symbols $X$ and $Y$ (resp. the information $\hat{V}$ provides about $X$).

The remainder of this paper is organized as follows.  In Section \ref{sec:probStatement} we formally define the problem and provide our main results.  In Section \ref{sec:proofs}, we discuss the properties of $d(\cdot)$ and provide the proofs of our main results.  Section \ref{sec:conc} delivers the conclusions and a brief discussion regarding further directions.

\begin{figure}[t]
\begin{center}
\includegraphics[scale=.32]{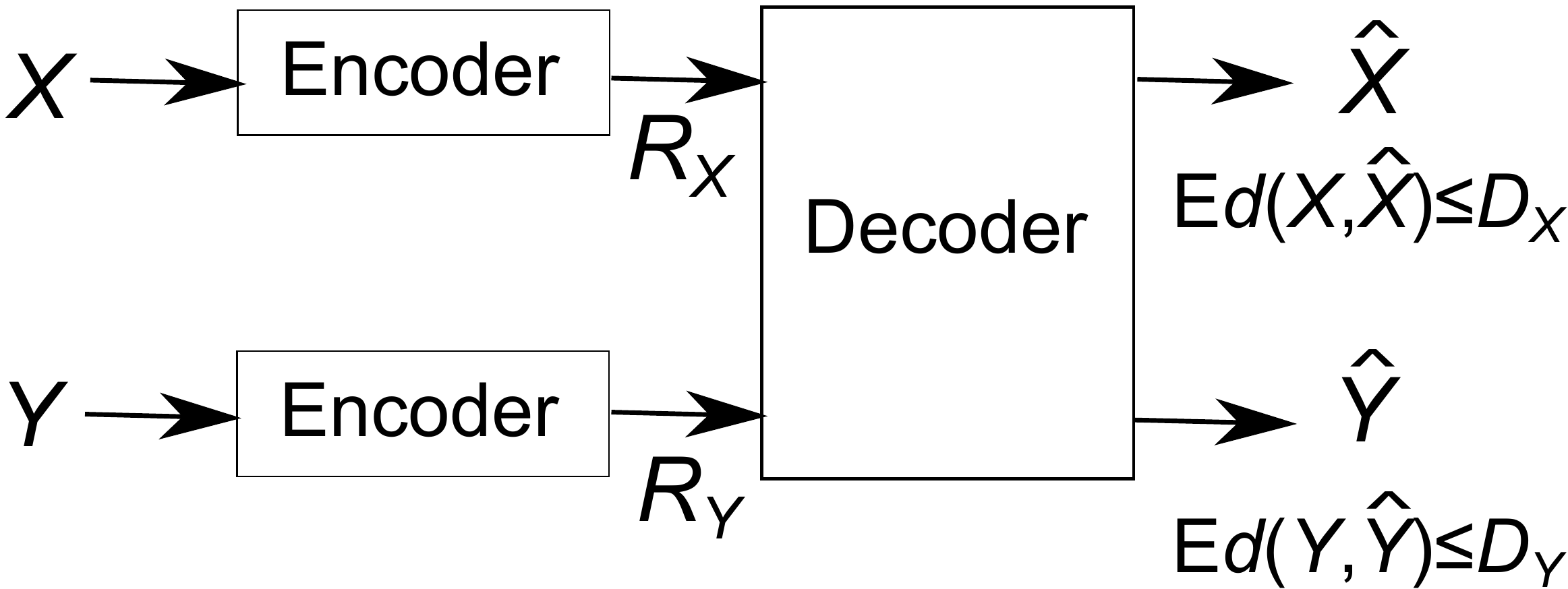}
\end{center}
\caption{Classical multiterminal source coding network.}
\label{fig:MTSC}
\end{figure}

\section{Problem Statement and Results} \label{sec:probStatement}
In this paper, we consider two cases of the lossy multiterminal source coding network presented in Fig. \ref{fig:JDXD}.

In the first case, we study the achievable rates $(R_x,R_y)$ subject to the joint distortion constraint
\begin{align*}
\mathbb{E}\left[d(X,Y,\hat{Z}) \right]\leq D,
\end{align*}
where $\hat{Z}$ is the joint reproduction symbol computed at the decoder from the messages $f_x$ and $f_y$ received from the $X$- and $Y$-encoders respectively.

In the second case, we study the achievable rates $(R_x,R_y)$ subject to a distortion constraint on $X$:
\begin{align*}
\mathbb{E}\left[d(X,\hat{V}) \right]\leq D_x,
\end{align*}
where $\hat{V}$ is the reproduction symbol computed at the decoder from the messages $f_x$ and $f_y$ received from the $X$- and $Y$-encoders respectively.  In this second case, there is no distortion constraint on $Y$.

\begin{definition}
To simplify terminology, we refer to the first and second cases described above as the Joint Distortion (JD) network and X-Distortion (XD) network respectively.
\end{definition}

\begin{figure}[htb]
\begin{center}
\includegraphics[scale=.32]{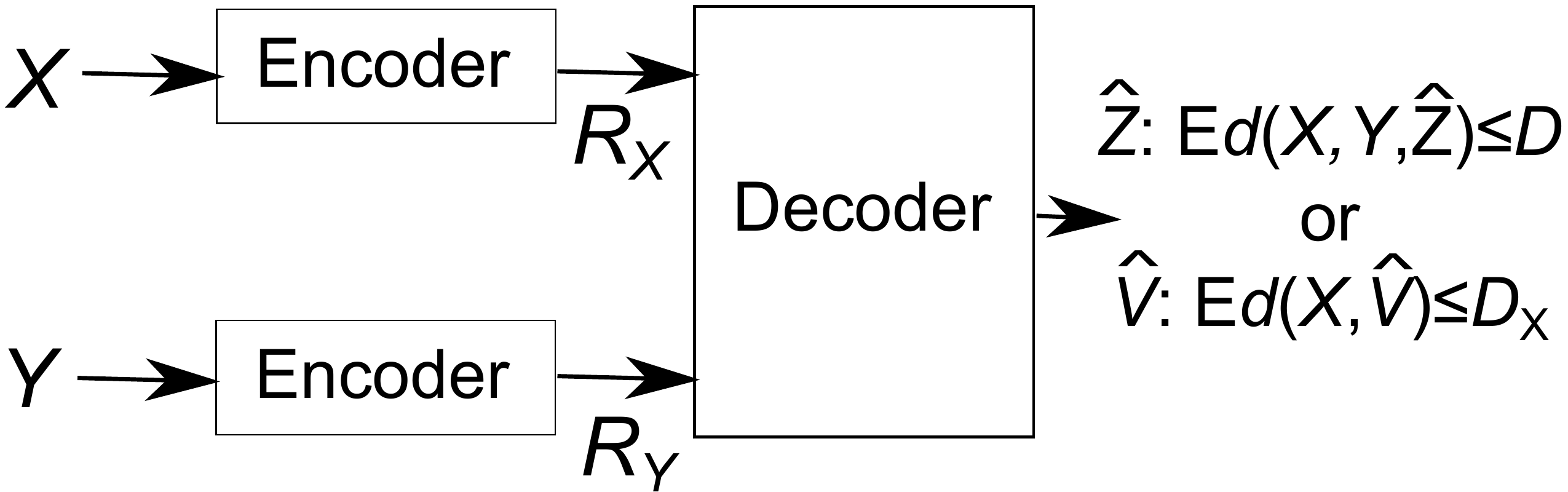}
\end{center}
\caption{The Joint Distortion (JD) and X-Distortion (XD) networks.}
\label{fig:JDXD}
\end{figure}

Formally, define the source alphabets as $\mathcal{X}=\{1,2,\dots,m\}$ and $\mathcal{Y}=\{1,2,\dots,\ell\}$.  We consider the discrete memoryless source sequences $X^n$ and $Y^n$ drawn i.i.d. according to the joint distribution $p(x,y)$.  Let $X^n$ be available at the $X$-encoder and $Y^n$ be available at the $Y$-encoder as depicted in Fig. \ref{fig:JDXD}. (We will informally refer to probability mass functions as distributions throughout this paper.)

For the case of joint distortion, we consider the reproduction alphabet $\hat{\mathcal{Z}}=\Delta_{m\times\ell}$, where $\Delta_k$ denotes the set of probability distributions on $k$ points.  In other words, for  $\hat{z}\in\hat{\mathcal{Z}}$, $\hat{z}=(q_{1,1},\dots,q_{m,\ell})$ where $q_{i,j}\geq0$ and $\sum_{i,j} q_{i,j}=1$.  With $\hat{z}$ defined in this way, it will be convenient to use the notation $\hat{z}(x,y)=q_{x,y}$ for $x\in\mathcal{X}, y\in\mathcal{Y}$.  Note that the restriction of the reproduction alphabet to the probability simplex places constraints on the function $\hat{z}(x,y)$.  For example, one cannot choose $\hat{z}(x,y)=x+y$.

Define the joint distortion measure $d:\mathcal{X}\times\mathcal{Y}\times\hat{\mathcal{Z}}\rightarrow \mathbb{R}^{+}$ by
\begin{align}
d(x,y,\hat{z})=\log\left( \frac{1}{\hat{z}(x,y)}\right),
\end{align}
and the corresponding distortion between the sequences $(x^n,y^n)$ and $\hat{z}^n$ as
\begin{align}
d(x^n,y^n,\hat{z}^n)=\frac{1}{n}\sum_{i=1}^n \log\left( \frac{1}{\hat{z}_i(x_i,y_i)}\right).\label{eqn:dx_expression}
\end{align}

As we will see in Section \ref{sec:proofs}, the distortion measure $d(\cdot)$ measures the amount of soft information that the reproduction symbols provide about the source symbols in such a way that the expected distortion can be described as an entropy.  For example, given the output from a discrete memoryless channel, the minimum distortion between the channel input and output is the conditional entropy.  For this reason, we refer to $d(\cdot)$ as an entropy-based distortion measure.

The function $d(\cdot)$ is a natural distortion measure for practical scenarios. A similar distortion measure has appeared previously in the image processing literature \cite{bib:dist} and in the study of the information bottleneck problem \cite{bib:Tis}. However, it does not appear to have been studied in the context of multiterminal source coding.

A $(2^{nR_x},2^{nR_y},n)$-rate distortion code for the JD network consists of encoding functions,
\begin{align*}
f_{x}&:\mathcal{X}^n\rightarrow\{1,2,\dots,2^{nR_x}\}\\
f_{y}&:\mathcal{Y}^n\rightarrow\{1,2,\dots,2^{nR_y}\},
\end{align*}
and a decoding function
\begin{align*}
g &: \{1,2,\dots,2^{nR_x}\} \times \{1,2,\dots,2^{nR_y}\}  \rightarrow\hat{\mathcal{Z}}^n.
\end{align*}

A vector $(R_x,R_y,D)$ with nonnegative components is achievable for the JD network if there exists a sequence of $(2^{nR_x},2^{nR_y},n)$-rate distortion codes satisfying
\begin{align*}
\lim_{n\rightarrow\infty}\mathbb{E}\left[d(X^n,Y^n,g(f_{x}(X^n),f_{y}(Y^n)) \right]\leq D.
\end{align*}

\begin{definition} The achievable rate distortion region, $\mathcal{R}$, for the JD network  is the closure of the set of all achievable vectors $(R_x,R_y,D)$.
\end{definition}

In a similar manner, we can also consider the case when there is only a distortion constraint on $X$ rather than a joint distortion constraint on $X,Y$.  For this, we consider the reproduction alphabet $\hat{\mathcal{V}}=\Delta_{m}$.  With $\hat{v}$ defined in this way, it will be convenient to use the notation $\hat{v}(x)=q_{x}$ for $x\in\mathcal{X}$.

We define the distortion measure $d_x:\mathcal{X}\times \hat{\mathcal{V}}\rightarrow \mathbb{R}^{+}$ by
\begin{align}
d_x(x,\hat{v})=\log\left( \frac{1}{\hat{v}(x)}\right),
\end{align}
and the corresponding distortion between the sequences $x^n$ and $\hat{v}^n$ as
\begin{align}
d_x(x^n,\hat{v}^n)=\frac{1}{n}\sum_{i=1}^n \log\left( \frac{1}{\hat{v}_i(x_i)}\right).\label{eqn:dx_expression}
\end{align}

Identical to the case for the JD network, we can define a $(2^{nR_x},2^{nR_y},n)$-rate distortion code for the XD network, with the exception that the range of the decoding function $g(\cdot)$ is the reproduction alphabet $\hat{\mathcal{V}}$.

A vector $(R_x,R_y,D_x)$ with nonnegative components is achievable for the XD network if there exists a sequence of $(2^{nR_x},2^{nR_y},n)$-rate distortion codes satisfying
\begin{align*}
\lim_{n\rightarrow\infty}\mathbb{E}\left[d_x(X^n,g(f_{x}(X^n),f_{y}(Y^n)) \right]\leq D_x.
\end{align*}

\begin{definition} The achievable rate distortion region, $\mathcal{R}_x$, for the XD network is the closure of the set of all achievable vectors $(R_x,R_y,D_x)$.
\end{definition}

Our main results are stated in the following theorems:
\begin{theorem} \label{thm:achievableRegion}
\begin{align*}
\mathcal{R}=\left\{ (R_x,R_y,D): \begin{array}{l}
\exists \delta_x,\delta_y \geq 0 \mbox{~such~that}\\
D\geq\delta_x+\delta_y\\
R_x+\delta_x \geq H(X|Y)\\
R_y+\delta_y \geq H(Y|X)\\
R_x+R_y+D \geq H(X,Y)
\end{array}
\right\}\end{align*}
\end{theorem}

\begin{theorem} \label{thm:achievableRegionX}
\begin{align*}
\mathcal{R}_x=\left\{ (R_x,R_y,D_x): \begin{array}{l}
R_x+D_x \geq H(X|U)\\
R_y \geq I(Y;U)\\
\mbox{for some distribution}\\
 p(x,y,u)=p_0(x,y)p(u|y),\\
\mbox{where~} |\mathcal{U}|\leq|\mathcal{Y}|+2.
\end{array}
\right\}\end{align*}
\end{theorem}

Since the distortion measure is reminiscent of discrete entropy, we can think of the units of distortion as ``bits" of distortion.  Thus, Theorem \ref{thm:achievableRegion} states that for every bit of distortion we allow for $X,Y$ jointly, we can remove exactly one bit of required rate from the constraints defining the Slepian-Wolf achievable rate region.  Indeed, we prove the theorem by demonstrating a correspondence between a modified Slepian-Wolf network and the multiterminal source coding problem in question.

Similarly, when we only consider a distortion constraint on $X$, Theorem \ref{thm:achievableRegionX} states that for every bit of distortion we tolerate, we can remove one bit of rate required by the $X$-encoder in the Ahlswede-K\"{o}rner-Wyner region.

The proofs of Theorems \ref{thm:achievableRegion} and \ref{thm:achievableRegionX} are given in the next section.

\section{Proofs}\label{sec:proofs}

We choose to prove Theorems \ref{thm:achievableRegion} and \ref{thm:achievableRegionX} by showing a correspondence between schemes that achieve a prescribed distortion constraint and the well-known lossless distributed source coding scheme of Slepian and Wolf, and the source coding with side-information scheme of Ahlswede, K\"{o}rner, and Wyner.  This provides a great deal of insight into how the various distortions are achieved.

In each case, the proof relies on a peculiar property of the distortion measure $d(\cdot)$.  Namely, the ability to convert expected distortions to entropies that are easily manipulated.  In the following subsection, we discuss the properties of the distortion measure $d(\cdot)$.

\subsection{Properties of $d(\cdot)$}
As stated above, one particularly useful property of $d(\cdot)$ is the ability to convert expected distortions to conditional entropies.  This is stated formally in the following lemma.

\begin{lemma}\label{lem:ExpectedDistortion}
Given any $U$ arbitrarily correlated with $(X^n,Y^n)$, the estimator $\hat{Z}^n[U]$ produces the expected distortion
\begin{align*}
\mathbb{E}\left[d(X^n,Y^n,\hat{Z}^n) \right] \geq \frac{1}{n}\sum_{i=1}^nH(X_i,Y_i|U).
\end{align*}
Moreover, this lower bound can be achieved by setting $\hat{z}_i[u](x,y):=\Pr\left(X_i=x,Y_i=y|U=u\right)$.
\end{lemma}
\begin{proof}
Given any $U$ arbitrarily correlated with $(X^n,Y^n)$, denote the reproduction of $(X^n,Y^n)$ from $U$ as $\hat{Z}^n[U]\in \hat{\mathcal{Z}}^n$.  By definition of the reproduction alphabet, we can consider the estimator $\hat{Z}^n[U]$ to be some probability distribution on $\mathcal{X}\times\mathcal{Y}$ conditioned on $U$.  Then, we obtain the following lower bound on the expected distortion conditioned on $U=u$:
\begin{align*}
\mathbb{E}&\,\left[d(X^n,Y^n,\hat{Z}^n) | U=u \right] \\&=
\frac{1}{n}\sum_{i=1}^n \sum_{x,y\in\mathcal{X}\times\mathcal{Y}} p_i(x,y|u) \log\left(\frac{1}{\hat{z}_i[u](x,y) }\right)\\
&= \frac{1}{n}\sum_{i=1}^n D\left(p_i(x,y|u)||\hat{z}_i[u](x,y)\right)+H(X_i,Y_i|U=u)\\
&\geq \frac{1}{n}\sum_{i=1}^n H(X_i,Y_i|U=u),
\end{align*}
where $p_i(x,y|u)=\Pr\left(X_i=x,Y_i=y|U=u \right)$ is the true conditional distribution.
Averaging both sides over all values of $U$, we obtain the desired result.  Note that the lower bound can always be achieved by setting $\hat{z}_i[u](x,y):=p_i(x,y|u)$.
\end{proof}

We now give two examples which illustrate the utility of the property stated in Lemma \ref{lem:ExpectedDistortion}.

\begin{example}\label{ex:WZ}
Consider the following theorem of Wyner and Ziv \cite{bib:WZ}:
\begin{theorem}
Let $(X,Y)$ be drawn i.i.d. and let $d(x,\hat{z})$ be given.  The rate distortion function with side information is
\begin{align*}
R_Y(D)=\min_{p(w|x)}\min_f I(X;W|Y)
\end{align*}
where the minimization is over all functions $f:\mathcal{Y}\times\mathcal{W}\rightarrow\hat{\mathcal{Z}}$ and conditional distributions $p(w|x)$ such that $\mathbb{E}\left[d(x,f(y,w)) \right]\leq D$.
\end{theorem}
For an arbitrary distortion measure, $R_Y(D)$ can be difficult to compute.  In light of Lemma \ref{lem:ExpectedDistortion} and its proof, we immediately see that:
\begin{align*}
R_Y(D)=H(X|Y)-D.
\end{align*}
\end{example}

\begin{example} \label{ex:RD}
As a corollary to the previous example, taking $Y=\emptyset$ we obtain the standard rate distortion function for a source $X^n$:
\begin{align*}
R(D)=H(X)-D.
\end{align*}
\end{example}

In both examples, we make the surprising observation that the distortion function $d(\cdot)$ yields a rate distortion function that is a multiple of the rate distortion function obtained using the ``erasure" distortion measure $d^{\infty}(\cdot)$ defined as follows:
\begin{align}
d^{\infty}(x,\hat{z})= \left\{\begin{array}{ll}
0 & \mbox{if $\hat{z}=x$}\\
\infty & \mbox{if $\hat{z} \neq x$ and $\hat{x} \neq e$}\\
1 & \mbox{if $\hat{z} = e$.} \end{array} \right.
\end{align}
This is somewhat counter-intuitive given the fact that an estimator is able to pass much more ``soft" information to the distortion measure $d(\cdot)$ compared to $d^{\infty}(\cdot)$.  It would be interesting to understand whether or not this relationship holds for general multiterminal networks, however this issue remains open.

\begin{definition}
We have defined $d(\cdot)$ to be a joint distortion measure on $\mathcal{X}\times\mathcal{Y}$, however it is possible to decompose it in a natural way.  We can define the marginal and conditional distortions for $X$ and $Y|X$ respectively by decomposing $\hat{z}_i[u](x,y)=\hat{z}_i(x|u)\hat{z}_i(y|x,u)$ (note the slight abuse of notation).  Thus, if the total expected distortion is less than $D$, we define the marginal and conditional distortions $D_x$, and $D_{y|x}$ as follows:
\begin{align*}
D &\geq \mathbb{E}\left[d(X^n,Y^n,\hat{Z}^n) \right] \\
&=\mathbb{E}\left[d_x(X^n,\hat{Z}^n) \right] + \mathbb{E}\left[d_{y|x}(Y^n,\hat{Z}^n) \right]\\
&:=D_x+D_{y|x}\\
&\geq \frac{1}{n} \sum_{i=1}^nH(X_i|U)+H(Y_i|U,X_i).
\end{align*}
In a complimentary manner, we can decompose the expected distortion into $D_y, D_{x|y}$ satisfying $D \geq D_y+D_{x|y}$.
\end{definition}

The definitions of expected total, marginal, and conditional distortion allow us to bound the number of sequences that are ``distortion-typical".  First, we require a result on peak distortion.

\begin{lemma} \label{lem:peakDistortion}
Suppose we have a sequence of $(2^{nR_x},2^{nR_y},n)$-rate distortion codes satisfying
\begin{align*}
\lim_{n\rightarrow\infty}\mathbb{E}\left[d(X^n,Y^n,g(f_{x}(X^n),f_{y}(Y^n)) \right]\leq D.
\end{align*}
For any $\epsilon>0$, $\Pr\left\{d(X^n,Y^n,\hat{Z}^n)>D+\epsilon\right\}<\epsilon$ for a sufficiently large blocklength $n$.
\end{lemma}
\begin{proof}
Suppose a length $n$ code satisfies the expected distortion constraint $\mathbb{E}\left[d(X^n,Y^n,\hat{Z}^n)\right]<D+\epsilon/2$.  By repeating the code $N$ times, we obtain $N$ i.i.d. realizations of $(X^n,Y^n,\hat{Z}^n)\sim p(X^n,Y^n,\hat{Z}^n)$.  By the weak law of large numbers:
\begin{align*}
\Pr\left\{d(X^{Nn},Y^{Nn},\hat{Z}^{Nn})>D+\epsilon\right\}<\epsilon
\end{align*}
for $N$ sufficiently large.
\end{proof}

Now, we take a closer look at the sets of source sequences that produce a given distortion.

\begin{lemma} \label{lem:DistortionSetSize}
Let $\mathcal{A}({\hat{z}^n})=\{(x^n,y^n):d(x^n,y^n,\hat{z}^n)\leq D+\epsilon\}$ for some $\epsilon>0$.  The size of $\mathcal{A}({\hat{z}^n})$ is bounded from above by $|\mathcal{A}({\hat{z}^n})|\leq 2^{n(D+2\epsilon)}$.
\end{lemma}
\begin{proof}
For each $(x^n,y^n)\in \mathcal{A}({\hat{z}^n})$, we can rearrange (\ref{eqn:dx_expression}) to obtain
\begin{align}
1\leq 2^{n(D+\epsilon)} \prod_{i=1}^n \hat{z}_i(x_i,y_i).\label{eqn:dx_expressionforContradiction}
\end{align}
By the definition of $\hat{z}^n$, observe that $\prod_{i=1}^n \hat{z}_i(x_i,y_i)$ is a valid probability measure on $\mathcal{X}^n\times\mathcal{Y}^n$.  Thus, for any subset $\mathcal{S}\subseteq \mathcal{X}^n\times\mathcal{Y}^n$, we have
\begin{align}
\sum_{(x^n,y^n)\in \mathcal{S}} \prod_{i=1}^n \hat{z}_i(x_i,y_i) \leq 1. \label{eqn:subset}
\end{align}
Combining (\ref{eqn:dx_expressionforContradiction}) and (\ref{eqn:subset}) gives the desired result:
\begin{align*}
|\mathcal{A}({\hat{z}^n})|&=\sum_{(x^n,y^n)\in \mathcal{A}({\hat{z}^n})}  1\\
&\leq \sum_{(x^n,y^n)\in \mathcal{A}({\hat{z}^n})}  2^{n(D+\epsilon)} \prod_{i=1}^n \hat{z}_i(x_i,y_i)\\
&\leq 2^{n(D+\epsilon)}.
\end{align*}
\end{proof}

We can also modify the previous result to include sequences which satisfy marginal and conditional distortion constraints.
\begin{lemma} \label{lem:ConditionalDistortionSetSize}
Let $\mathcal{A}_x({\hat{z}^n})=\{x^n:d_x(x^n,\hat{z}^n)\leq D_x+\epsilon\}$ and $\mathcal{A}_{y|x}({\hat{z}^n})=\{y^n:d_{y|x}(y^n,\hat{z}^n)\leq D_{y|x}+\epsilon\}$ for some $\epsilon>0$.  The sizes of these sets are bounded as follows:
\begin{align*}
|\mathcal{A}_x({\hat{z}^n})|&\leq 2^{n(D_x+2\epsilon)}, \quad \mbox{and}\\
|\mathcal{A}_{y|x}({\hat{z}^n})|&\leq 2^{n(D_{y|x}+2\epsilon)}
\end{align*}
for sufficiently large $n$.  Symmetric statements hold for $\mathcal{A}_y({\hat{z}^n})$ and $\mathcal{A}_{x|y}({\hat{z}^n})$.
\end{lemma}
\begin{proof}
The proof is nearly identical to that of Lemma \ref{lem:DistortionSetSize} and is therefore omitted.
\end{proof}

\subsection{Proof of Theorem \ref{thm:achievableRegion} }
As mentioned previously, we prove Theorem \ref{thm:achievableRegion} by demonstrating a correspondence between the JD network with a joint distortion constraint and a Slepian-Wolf network.  To this end, we now define a modified Slepian-Wolf code.  Essentially the code splits the rates of each user into two parts.  We refer to this network as the Split-Message Slepian-Wolf (SMSW) network.

A $(2^{nR_x},2^{nR_y},2^{nR_{0x}},2^{nR_{0y}},n)$-SW (Slepian-Wolf) code for the SMSW network consists of encoding functions,
\begin{align*}
\phi_{x}&:\mathcal{X}^n\rightarrow\{1,2,\dots,2^{nR_x}\}\\
\phi_{y}&:\mathcal{Y}^n\rightarrow\{1,2,\dots,2^{nR_y}\}\\
\psi_x &:\mathcal{X}^n\rightarrow \{1,2,\dots,2^{n\delta_1}\}\\
\psi_y &:\mathcal{Y}^n\rightarrow \{1,2,\dots,2^{n\delta_2}\},
\end{align*}
and a decoding function
\begin{align*}
\chi &: [2^{nR_x}] \times [2^{nR_y}] \times [2^{n\delta_1}] \times [2^{n\delta_2}] \rightarrow \mathcal{X}^n\times \mathcal{Y}^n.
\end{align*}

A vector $(R_x,R_y,\delta_1,\delta_2)$ with nonnegative components is achievable for the SMSW network if there exists a sequence of $(2^{nR_x},2^{nR_y},2^{n\delta_1},2^{n\delta_2},n)$-SW codes satisfying
\begin{align*}
\lim_{n\rightarrow\infty} \Pr\left\{ (X^n,Y^n) \neq \chi(\phi_x,\phi_y,\psi) \right\}=0.
\end{align*}

\begin{definition} The achievable region, $\mathcal{R}_{SW}$, for the SMSW network  is the closure of the set of all achievable vectors $(R_x,R_y,\delta_1,\delta_2)$.
\end{definition}

\begin{theorem}[\cite{bib:SW}]\label{thm:SW}
The achievable rate region $\mathcal{R}_{SW}$ consists of all rate tuples $(R_x,R_y,\delta_1,\delta_2)$ satisfying
\begin{align*}
R_x+\delta_1&\geq H(X|Y)\\
R_y+\delta_2&\geq H(Y|X)\\
R_x+R_y+\delta_1+\delta_2 &\geq H(X,Y).\\
\end{align*}
\end{theorem}

\begin{claim}
If $(R_x,R_y,D)$ is an achievable rate-distortion vector for the JD network, then $(R_x,R_y,\delta_1,\delta_2)$ is an achievable rate vector for the SMSW network for some $\delta_1,\delta_2\geq 0$ such that $\delta_1+\delta_2\leq D$.
\end{claim}
\begin{proof}
Suppose we have a sequence of $(2^{nR_x},2^{nR_y},n)$-rate distortion codes satisfying
\begin{align*}
\lim_{n\rightarrow\infty}\mathbb{E}\left[d(X^n,Y^n,g(f_{x}(X^n),f_{y}(Y^n)) \right]\leq D.
\end{align*}  From these codes, we will construct a sequence of $(2^{nR_x},2^{nR_y},2^{n\delta_1^{(n)}},2^{n\delta_2^{(n)}},n)$-SW codes satisfying $\lim_{n\rightarrow\infty}\delta_1^{(n)}+\delta_2^{(n)}  \leq D$ and
\begin{align*}
\lim_{n\rightarrow\infty} \Pr\left\{ (X^n,Y^n) \neq \chi(\phi_x,\phi_y,\psi_x,\psi_y) \right\}=0.
\end{align*}

The encoding procedure is almost identical to the rate distortion encoding procedure.  In particular, set $\phi_x(X^n)=f_x(X^n)$ and $\phi_y(Y^n)=f_y(Y^n)$.  Decompose the expected joint distortion into the marginal and conditional distortions $D_x$,$D_{y|x}$ which must satisfy $D_x+D_{y|x}\leq D+\epsilon$ by definition.

Define the remaining encoding functions $\psi_x, \psi_y$ as follows:  Bin the $X^n$ sequences randomly into $2^{n(D_x+3\epsilon)}$ bins and, upon observing the source sequence $X^n$, set $\psi_x(X^n)=b_x(X^n)$ (where $b_x(X^n)$ is the bin index of $X^n$).  Similarly, bin the $Y^n$ sequences randomly into $2^{n(D_{y|x}+3\epsilon)}$ bins and, upon observing the source sequence $Y^n$, set $\psi_y(Y^n)=b_y(Y^n)$ (where $b_y(Y^n)$ is the bin index of $Y^n$).

The decoder finds the unique $\hat{X}^n$ in bin $b_x(X^n)$ satisfying $d_x(\hat{X}^n,\hat{Z}^n)<D_x+\epsilon$.  If $\hat{X}^n\neq X^n$, an error occurs.  Upon successfully recovering $X^n=\hat{X}^n$, the decoder finds the unique $\hat{Y}^n$ such that $d_{y|x}(\hat{Y}^n,\hat{Z}^n)<D_{y|x}+\epsilon$.  If $\hat{Y}^n\neq Y^n$, an error occurs.

The various sources of error are the following:
\begin{enumerate}
\item An error occurs if $d_x(X^n,g(\phi_x,\phi_y))>D_x+\epsilon$ or $d_{y|x}(Y^n,g(\phi_x,\phi_y))>D_{y|x}+\epsilon$.  By Lemma \ref{lem:peakDistortion}, this type of error occurs with probability at most $\epsilon$.
\item An error occurs if there is some other $\tilde{X}^n \neq X^n$ in bin $b_x(X^n)$ satisfying $d_x(\tilde{X}^n,g(\phi_x,\phi_y))<D_x+\epsilon$.  By Lemma \ref{lem:DistortionSetSize} and the observation that that $\Pr\left\{\tilde{X}^n\in \mbox{~bin~} b_x(X^n) \right\}=2^{-n(D_x+3\epsilon)}$, this type of error occurs with arbitrarily small probability.
\item An error occurs if there is some other $\tilde{Y}^n \neq Y^n$ in bin $b_y(Y^n)$ satisfying $d_{y|x}(\tilde{Y}^n,g(\phi_x,\phi_y))<D_{y|x}+\epsilon$.  By Lemma \ref{lem:DistortionSetSize} and the observation that that $\Pr\left\{\tilde{Y}^n\in \mbox{~bin~} b_y(Y^n) \right\}=2^{-n(D_{y|x}+3\epsilon)}$, this type of error is also small.
\end{enumerate}
At this point the proof is essentially complete, but there is a minor technical difficulty dealing with the sequences $\left\{\delta_1^{(n)},\delta_2^{(n)} \right\}_{n=1}^{\infty}$ corresponding to the sequences of marginal and conditional distortions computed from $\hat{Z}^n$ for each $n$.  We require that there exists some $\delta_1$ such that $\delta_1^{(n)}\rightarrow \delta_1$ and similarly for the sequence of $\delta_2^{(n)}$'s.  However, since $[0,D+\epsilon]\times[0,D+\epsilon]$ is compact, we can find a convergent subsequence $\left\{\delta_1^{(n_j)},\delta_2^{(n_j)} \right\}_{j=1}^{\infty}$ so that the desired limits exist.
\end{proof}

\begin{claim}
If $(R_x,R_y,\delta_1,\delta_2)$ is an achievable rate vector for the SMSW network, then $(R_x,R_y,\delta_1+\delta_2)$ is an achievable rate distortion vector for the JD network.
\end{claim}
\begin{proof}
By Theorem \ref{thm:SW}, we must have:
\begin{align*}
R_x&\geq H(X|Y)-\delta_1\\
R_y&\geq H(Y|X)-\delta_2\\
R_x+R_y&\geq H(X,Y)-\delta_1-\delta_2.
\end{align*}
Let $D=\delta_1+\delta_2$.  For fixed $\delta_1,\delta_2$, any nontrivial $(R_x,R_y)$ pair in this region can be achieved by an appropriate time-sharing scheme between the two points
\begin{align*}
P_1&=\left( \max\left\{H(X|Y)-D,0\right\}, \right. \\
& ~~~~~\left. \min\left\{H(Y),H(Y)-\left(D-H(X|Y) \right) \right\} \right), \mbox{~and}\\
P_2&=\left( \min\left\{H(X),H(X)-\left(D-H(Y|X) \right)\right\}, \right. \\
& ~~~~~\left. \max\left\{H(Y|X)-D,0 \right\} \right).
\end{align*}
By the results given in Examples \ref{ex:WZ} and \ref{ex:RD}, point $P_1$ allows $X,Y$ to be recovered with distortion $D$. Symmetrically, point $P_2$ allows $X,Y$ to be recovered with distortion $D$.  Thus, using the appropriate time-sharing scheme to generate average rates $(R_x,R_y)$, we can create a sequence of rate distortion codes that achieve the point $(R_x,R_y,D)$ for the JD network.
\end{proof}

\subsection{Proof of Theorem \ref{thm:achievableRegionX} }
The proof of Theorem \ref{thm:achievableRegionX} is similar in spirit to the proof of Theorem \ref{thm:achievableRegion} and has therefore been moved to the appendix.  The key difference between the proofs is that, instead of showing a correspondence between $\mathcal{R}$ and the SMSW achievable rate region, we show a correspondence between $\mathcal{R}_X$ and the Ahlswede-K\"{o}rner-Wyner achievable rate region.

\section{Conclusion}\label{sec:conc}  In this paper, we gave the rate distortion regions for two different multiterminal networks subject to distortion constraints using the entropy distortion measure.  In the case of the Joint Distortion and X-Distortion networks, we observed that any point in the rate distortion region can be achieved by timesharing between points in the SMSW region and the Ahlswede-K\"{o}rner-Wyner regions respectively.  Perhaps this is an indication that the rate distortion region for more general multiterminal source networks (subject to distortion constraints using the entropy distortion measure) can be characterized by simpler source networks for which achievable rate regions are known.  This is one potential direction for future investigation.

\section*{Appendix}
This appendix contains a sketch of the proof for Theorem \ref{thm:achievableRegionX}.
\begin{claim}
If $(R_x,R_y,D_x)$ is an achievable rate-distortion vector for the XD network, then $(R_x+D_x,R_y)$ is an achievable rate vector for the source coding with side information problem.
\end{claim}
\begin{proof}
Suppose we have a sequence of $(2^{nR_x},2^{nR_y},n)$-rate distortion codes satisfying
\begin{align*}
\lim_{n\rightarrow\infty}\mathbb{E}\left[d_x(X^n,g(f_{x}(X^n),f_{y}(Y^n)) \right]\leq D_x.
\end{align*}

The basic idea is to let the $X$-encoder send $f_x(X^n)$ (requiring rate $R_x$) and have the $Y$-encoder send $f_y(Y^n)$ (requiring rate $R_y$).  By Lemma \ref{lem:ConditionalDistortionSetSize}, the number of $X^n$ sequences that lie in $\mathcal{A}_x(\hat{v}^n)$ is less than $2^{n(D_x+2\epsilon)}$.  Therefore, if the $X$-encoder performs a random binning of the $X^n$ sequences into $2^{n(D_x+3\epsilon)}$ and sends the bin index corresponding to the observed sequence $X^n$ (incurring an additional rate of $D_x+3\epsilon$), the decoder can recover $X^n$ losslessly with high probability.
\end{proof}

\begin{claim}
If $(R_x+D_x,R_y)$ is an achievable rate-distortion vector for the source coding with side information network, then $(R_x,R_y,D_x)$ is an achievable rate distortion vector for the XD network.
\end{claim}
\begin{proof}
Since $(R_x+D_x,R_y)$ is an achievable rate vector, there exists some conditional distribution $p(u|y)$ so that $R_x+D_x \geq H(X|U)$ and $R_y \geq I(Y;U)$.  WLOG, reduce $R_x$ and $R_y$ if necessary so that $R_x+D_x = H(X|U)$ and $R_y = I(Y;U)$.  Now, we construct a sequence of codes that achieve that point in the standard way.  In particular, generate $2^{n(R_y+\epsilon)}$ different $U^n$ sequences independently i.i.d. according to $p(u)$. Upon observing $Y^n$, the $Y$-encoder finds a jointly typical $U^n$ and sends the corresponding index to the decoder. At the $X$-encoder, bin the $X^n$ sequences into $2^{n(R_x+D_x+2\epsilon)}$ bins and, upon observing the source sequence $X^n$, send the corresponding bin index to the decoder.  With high probability, the decoder can reconstruct $X^n$ losslessly.

From this sequence of codes, we can construct a sequence of rate distortion codes that achieve the point $(R_x,R_y,D_x)$ as follows.  At the $X$-encoder, employ the following time-sharing scheme:
\begin{enumerate}
\item Use the lossless code described above with probability $(1-D_x/H(X|U))$.  In this case, the distortion on $X$ can be made arbitrarily small.  Note that we can assume w.l.o.g. that $D_x<H(X|U)$ since distortion $D_x=H(X|U)$ can be achieved when the decoder only receives the sequence $U^n$.
\item With probability $D_x/H(X|U)$, the $X$-encoder sends nothing, while the $Y$-encoder continues to send $U^n$.  In this case, the distortion on $X$ is $H(X|U)$.
\end{enumerate}

Averaging over the two strategies, we obtain a sequence of rate distortion codes that achieve the rate distortion triple $(R_x,R_y,D_x)$.
\end{proof}

\newpage



%

\end{document}